\documentclass[graybox]{svmult}

\usepackage{helvet}         
\usepackage{courier}        
\usepackage{type1cm}        

\usepackage{makeidx}         
\usepackage{graphicx}        
                             
\usepackage{multicol}        
\usepackage[bottom]{footmisc}
\usepackage[mathscr]{eucal}
\usepackage{amssymb}
\usepackage{bbold}

\newcommand{\be}{\begin{eqnarray}}
\newcommand{\ee}{\end{eqnarray}}
\newcommand{\nn}{\nonumber \\}
\newcommand{\lb}{\label}
\newcommand{\p}[1]{(\ref{#1})}
\newtheorem{thm}{Theorem}
\newtheorem{defi}{Definition}

\begin{document}

\title*{Comments on the Newlander-Nirenberg theorem}

\author{A.V.~Smilga}

\institute{A.V.~Smilga \at  University of Nantes, \email{Smilga@subatech.in2p3.fr}}

\maketitle

\abstract{The Newlander-Nirenberg theorem says that a necessary and sufficient condition for the complex coordinates associated with a given almost complex structure tensor $I_M{}^N$ to exist is the vanishing of the Nijenhuis tensor ${\cal N}_{MN}{}^K$. In the first part of the paper, we give a heuristic but very simple proof of this fact. In the second part, we discuss a supersymmetric interpretation of this theorem. {\it (i)} The condition ${\cal N}_{MN}{}^K = 0$ is necessary for  certain ${\cal N}=1$ supersymmetric mechanical sigma models to enjoy ${\cal N}=2$ supersymmetry. {\it (ii)} The sufficiency of this condition for the existence of complex coordinates implies that the representation of the supersymmetry algebra realized by the superfields associated with  all the real coordinates and their superpartners can be presented as a direct sum of $d$ irreducible representations ($d$ is the complex dimension of the manifold)}

\section{Introduction}
\label{sec:1}

Since 1982, we know 
that many well-known structures of differential geometry, such as the de Rham complex, 
allow for a supersymmetric interpretation \cite{Wit-Morse}.
For any manifold, one can define a certain supersymmetric quantum mechanical model. 
The dynamical time-dependent variables of this model include the coordinates and their Grassmann-valued superpartners.

Supersymmetric language is very useful. Besides giving a new unexpected interpretation 
of known mathematical facts, it allows one to derive many {\it new} nontrivial results, which are difficult 
to derive in a traditional way. I give here only one example. The so-called HKT 
manifolds were first discovered by supersymmetric methods \cite{HKT} and only then they attracted the attention of pure mathematicians who gave their traditional description \cite{HKT-math}. The full classification of HKT metrics was also recently constructed using supersymmetric tools \cite{DI,HKT-nonlin}.

Supersymmetry is a standard method to study geometrical properties of the manifolds used by 
``physicists'' (I've put here the quotation marks because we are talking in this case
 about the scholars who may have studied physics at university, but who are now solving pure mathematical 
 problems without much relationship to the physical world) in the papers published in the hep-th section
  of the {\it arXiv}. On the other hand, pure mathematicians are reluctant to use it, preferring traditional methods. 

It is an unfortunate fact of our life that a large gap exists between the  two communities. The languages in which the papers 
are written and the ways of thinking derived from these languages are often very different, to the extent that mathematicians 
and physicists do not often understand each other, even though the subject of their studies could be practically identical.

That is exactly the reason by which I've decided to write this methodical paper.
Its second half is mainly addressed to mathematicians who might be curious to learn that a certain well-known mathematical fact admits 
an unexpected interpretation in the supersymmetry framework. And its first half is addressed to physicists who might have heard 
about the NN theorem, but probably do not know how it is proven. Indeed, its rigourous mathematical proof is not so trivial. So I give here a heuristic but simple reasoning, presenting  the solution to the equations \p{dz/dx-multi} as the perturbative series over a deviation of the complex structure tensor $I_M{}^N(x)$ from its flat form. This reasoning might be upgraded to a rigourous proof if the convergence of this series is proven.

\section{Geometry}
\label{sec:2}

\subsection{Preliminaries}
 
 \begin{defi}
 {\em A complex manifold is a manifold of even dimension $D= 2d$ which can be represented as a union of several overlapping charts such that:
 \begin{enumerate}
 \item Each chart is homeomorphic to $\mathbb{R}^D$.
 \item In each chart, one can define complex coordinates $z^n$. 
  \item In a region where two charts overlap, the coordinates $z^n$ in one chart and the coordinates $w^n$ in another chart are related by {\it holomorphic} transition  functions $z^n = f^n(w^m)$. 
 \end{enumerate} }
 \end{defi}
 
 \begin{defi}
 {\em A Hermitian manifold is a complex manifold endowed by Hermitian metric
 \be
 \lb{Herm}
 ds^2 \ =\ 2h_{n\bar m}\, dz^n d\bar z^{\bar m} 
  \ee
  with $\overline{h_{n\bar m}} = h_{m\bar n }$.}
   \end{defi}
   
   The factor 2 was introduced here for further conveniences---to make contact with the standard normalization in \p{L-compl-N2} and \p{L-real-N2}. 
   Mathematicians sometimes consider manifolds not endowed with the metric. In particular, the NN theorem can be formulated and proven without using the notion of metric. But we are interested in a supersymmetric interpretation of the NN theorem, and we can only give it if the Hermitian metric \p{Herm} is defined. Thus, its existence will be assumed. 
 
 An interesting and important fact is that one can describe 
complex manifolds without explicitly introducing
       complex charts, but working exclusively in the real terms.\footnote{It is convenient---especially, for supersymmetric applications---but is
not {\it necessary}. For example, the popular textbook \cite{Kodaira} uses only complex but not real description.} To this end, we introduce first the notion of an {\it almost}
 complex manifold:
   \begin{defi}
       \label{defi-almost}
       {\em An {\it almost complex manifold} is a manifold of even dimension $D$ endowed with a 
globally defined tensor field $I_{MN}$ satisfying the properties
       {\it (i)} $ I_{MN} = -I_{NM}$ and {\it (ii)} $I_M{}^N I_N{}^P = -\delta_M^P$. 
        The tensor $I_M{}^N$ is called the {\it almost 
       complex structure}}.
       \end{defi}

  To understand why a {\it real} tensor is called {\it complex} structure, consider first the 
simplest possible example---flat 2-dimensional 
Euclidean space. It can be parametrized by the 
real Cartesian coordinates $x^1, x^2$ or by the 
complex coordinate $z = (x^1 + ix^2)/\sqrt{2}$. An obvious 
relation $\partial z/\partial x^2 = i \partial z/\partial x^1$ 
holds, which can also be presented in the form
           \be
            \lb{dz/dx-2D}
          \frac {\partial z } {\partial x_A} - i  
\varepsilon_{AB} \, \frac {\partial z } {\partial x_B}  \ =\ 0 
           \ee
with 
 \be
 \lb{eps-def}
 \varepsilon = \left( \begin{array}{cc}  0&-1 \\ 1&0 \end{array} \right) \, .  
   \ee 
     The tensor $\varepsilon_{AB}$ 
satisfies both conditions in the definition 
above and {\it is} the complex structure in this case. Note that the property  \p{dz/dx-2D} holds not only for $z$, but for any holomorphic function
$f(z)$. In the latter case, the real and imaginary parts of \p{dz/dx-2D}
 are none other than the Cauchy-Riemann conditions.
 
 If a 2-dimensional manifold is not flat, $I_M{}^N$ may have a little bit more complicated form, but its tangent space projection $I_{AB}= I_{MN}e^M_Ae^N_B$ coincides with the matrix $\epsilon$ or probably with $-\epsilon$. Indeed, an antisymmetric $2\times 2$ matrix whose square is 
 $-{\mathbb 1}$ coincides with \p{eps-def} up to a sign. It describes rotations by $\pi/2$ or by $-\pi/2$. 
 
  In the general multidimensional case, one can prove a simple theorem:
    \begin{thm}
     {\em Take a tensor $I_M{}^N$ 
satisfying the conditions above. With a proper choice of the vielbeins $e^M_A$ (with a proper choice of the orthonormal base in the tangent space), its tangent 
space projection can be
     brought to the canonical form}
 \be
\lb{diag-eps}
    I_{AB} \ =\ {\rm diag} \, (\varepsilon, \ldots, \varepsilon) \, .
     \ee
 \end{thm}
 \begin{proof}    
To construct an orthonormal base in the tangent space $E$ where the complex structure acquires the form
\p{diag-eps}, we start with choosing in $E$ an arbitrary unit vector $e_1$. It follows from
 $I = -I^T$
and $I^2 = -\mathbb{1}$ that the vector $e_2 = Ie_1$ has also unit length and is orthogonal to 
$e_1$. Obviously, $Ie_2 = I^2 e_1 = - e_1$. Consider the subspace $E^* \subset E$ that
is orthogonal to $e_1$ and $e_2$. If it is not empty, choose there an arbitrary unit vector $f_1$ and consider  $f_2 = If_1$.
 One can easily see that $f_2$ also belongs to $E^*$.  Now consider the subspace $E^{**} \subset  E^* \subset E$ that is orthogonal to
$e_{1,2}, f_{1,2}$ and, if $E^{**}$ is not empty, repeat the procedure. 
We arrive at the matrix \p{diag-eps}.  \label{canon}
    \end{proof}

 Now consider the equation system
      \be
      \lb{dz/dx-multi}
        \frac {\partial z^n } {\partial x^M} - i  I_M{}^N  \, \frac {\partial z^n } 
{\partial x^N}  \ =\ 0 
           \ee          
If not only $I_{AB}$, but also $I_M{}^N$ has the form \p{diag-eps}, solutions to \p{dz/dx-multi} can be easily found. A simple set of $d$ independent solutions is
\be
  \lb{z0-sol}
  z_{(0)}^1 \ =\ \frac   {x^1 + ix^2} {\sqrt{2}}\,, \quad  z_{(0)}^2 \ =\ \frac {x^3 + ix^4} {\sqrt{2}} \,, \ldots \,
   \ee
   or any set of $d$ non-degenerate analytic functions of $ z_{(0)}^n$.
   
   In a generic case, the solutions to \p{dz/dx-multi} are more complicated. Moreover, they do not always exist. The conditions under which they do, is the content of the NN theorem to be proven in the next section. For the time being, we will prove that
      \begin{thm}
             {\em If the equation system \p{dz/dx-multi} has $d$ independent solutions,
 the manifold is complex. Its metric is Hermitian. }
             \end{thm} 
             Actually, as follows from Theorem 3 below, it is sufficient to require the existence of only one such solution. 
\begin{proof}
We will show first that the metric has a Hermitian form (i.e. the components $g^{nm}$ etc vanish) 
Let us  trade $x^M$  for $(z^n \, ,\ \bar z^{\bar n})$ 
and write 
$$g^{nm} \ =\ \frac {\partial z^n}{\partial x^M} \frac {\partial z^m}{\partial x^N}\, g^{MN} = 
i I_M{}^P \frac {\partial z^n}{\partial x^P} \frac {\partial z^m}{\partial x^N}\, g^{MN} = i I^{NP} \frac {\partial z^n}{\partial x^P} \frac {\partial z^m}{\partial x^N} = 0 $$
by symmetry considerations.
The vanishing of $g^{\bar n \bar m}$ follows from the same argument. The properties $g^{\bar n \bar m} = g^{nm} = 0$ imply also the vanishing of the 
components $g_{nm}$ and  $g_{\bar n \bar m}$ of the inverse tensor.

Next, we need to show that the transition functions between two overlapping charts with the coordinates $(z^n, \bar z^{\bar n})$ and 
$(w^m, \bar w^{\bar m})$ are holomorhic. 
To this end, we express, using \p{dz/dx-multi}, $I_M{}^N$ in the complex frame,
             \be
             \lb{I-compl}
 I_m{}^n   =  I_M{}^N \frac {\partial z^n}{\partial x^N} \frac {\partial x^M}{\partial z^m} = 
 -i \frac {\partial z^n}{\partial x^M} \frac {\partial x^M}{\partial z^m} \ =\   
 -i \delta_m^n\,, \nn
  I_{\bar m}{}^{\bar n} \ =\   i \delta_{\bar m}^{\bar n} \,,\ \ \  
              I_m{}^{\bar n} \ =\  I_{\bar m}{}^n \ =\ 0 
               \ee 
and consider the transformation of the tensor \p{I-compl} from one chart to another. Knowing that  $I$  keeps the form \p{I-compl} after this transformation, one can derive that $\partial w^m/\partial \bar z^{\bar n} = 0$.
                               \end{proof}

\subsection{NN theorem}
Not wishing to plunge into not relevant for us details, we assume that the manifold and all its structures are real analytic (can be expanded in the Taylor series). The traditional proof of the NN theorem in \cite{NN} assumes the existence of $D=2d$ derivatives. H\"ormander proved that it is sifficient to require the existence of the first derivative \cite{Hormand}.

Introduce the  object 
\be
       \lb{Nijen-def}
     {\cal N}_{MN}{}^K =   \partial_{[M} I_{N]}{}^K -   
I_M{}^P  I_N{}^Q \partial_{[P}  I_{Q]}{}^K  \,.
       \ee
        It is a tensor, in spite of the presence of the ordinary rather than covariant derivatives. This is so because one can replace the ordinary derivatives by the covariant ones---the  terms  involving the Christoffel symbols cancel out in this case. Using a sloppy language, we will call the L.H.S. of Eq. \p{Nijen-def}  the {\it Nijenhuis tensor}.\footnote{A conventional definition of the Nijenhuis tensor is a little bit different: 
     \be
      \lb{conven-Nijen} 
      {\cal N}_{MN}{}^{K(\rm conventional)} = I_M{}^P {\cal N}_{PN}{}^{K (\rm this\ paper)}\ =\ I_M{}^P \partial_{[P} I_{N]}{}^K  + I_N{}^P \partial_{[M} I_{P]}{}^K 
       \ee 
       (the last equality holds due to  $I^2 = -\mathbb{1}$).}
      We will do so because the object \p{Nijen-def} has a more transparent structure, and it is this combination that will directly appear later in \p{comm-D}. 
      
      The NN theorem says that
      
 \begin{thm} 
\lb{NN1}
      {\em \cite{NN} The complex coordinates satisfying the condition 
\p{dz/dx-multi} can be introduced and the manifold is complex iff
      the condition
       \be
       \lb{Nijen}
     {\cal N}_{MN}{}^K    \ =\ 0  
       \ee
holds.}
 \end{thm}

   \begin{proof}
\lb{huis}

\ 

 {\bf Necessity.}
 Represent the system \p{dz/dx-multi} as
${\cal D}_M z^n = 0$ with 
\be
{\cal D}_M = \partial_M  - i I_M{}^N \partial_N\,.
\ee
For self-consistency, the conditions $[{\cal D}_M, {\cal D}_N] z^n = 0$ should also hold. We derive
 \be
\lb{comm-D}
&&[{\cal D}_M, {\cal D}_N] z^n \ =\ \left[-i \partial_{[M} I_{N]}{}^Q - I_{[M}{}^P (\partial_P I_{N]}{}^Q)
\right] \partial_Q z^n \nn
 &=& \left[-i   \partial_{[M} I_{N]}{}^K - i  I_{[M}{}^P (\partial_P I_{N]}{}^Q)
I_{Q}{}^K \right] \partial_K z^n - I_{[M}{}^P (\partial_P I_{N]}{}^Q) {\cal D}_Q z^n \,.
 \ee
 Bearing in mind that ${\cal D}_Q z^n = 0$, the last term in the R.H.S. vanishes. The middle term can be transformed by flipping the derivative,
 $  (\partial_P I_{N}{}^Q) I_{Q}{}^K = -
 I_{N}{}^Q  \partial_P I_{Q}{}^K$ (this holds due to $I^2 = -\mathbb{1}$), and we finally obtain 
  \be 
[{\cal D}_M, {\cal D}_N] z^n \ =\ -i  {\cal N}_{MN}{}^K \partial_K z^n\,.
  \ee
 For this to vanish, the tensor ${\cal N}_{MN}{}^K$ should also vanish (to see this, choose the real coordinates as the real and imaginary parts of $z^n$).

 \vspace{2mm}
 
 {\bf Sufficiency.} This part of the theorem [the proof of existence of the solution to the system \p{dz/dx-multi} under the condition \p{Nijen}] is more diffucult. Well, it might be not so difficult for the mathematicians in the case when the complex structures $I_M{}^N$ represent analytic functions of the coordinates. Then the sufficiency of the conditions $[{\cal D}_M, {\cal D}_N] = {\cal K}_{MN}{}^Q \, {\cal D}_Q$ for the equation system ${\cal D}_M z^n = 0$ to have a solution is a corollary of the classical Frobenius theorem \cite{Frobenius}. We will give here instead a  heuristic  proof of the NN theorem using ``physical'' language. This proof will elucidate the meaning of the constraint \p{Nijen}. Its linearized version is similar in spirit to multidimensional  Cauchy-Riemann conditions.
 \begin{itemize}
 \item Let the complex structure $I_M{}^N$ has a canonic form \p{diag-eps}. Then the solutions to \p{dz/dx-multi} exist, and one of the solution is given by
 \p{z0-sol}.

   Suppose now that the complex structure does not coincide with $(I_0)_M{}^N \ = {\rm diag}(\varepsilon, \ldots, \varepsilon)$, 
   but is close to it: $I = I_0 + \Delta$, \ $ \Delta \ll 1$. As a first step in the proof, we will show that, after such an infinitesimal deformation, solutions to \p{dz/dx-multi} still exist.
   
   \item Let us first do so in the simplest case $D=2$. Then the condition \p{Nijen} is fulfilled identically. The condition $I^2 = - \mathbb{1}$ means that  $\{\Delta, I_0\} = 0$, which is so iff \footnote{In physical notation, $\Delta = \alpha \sigma^1 + \beta \sigma^3$, where $\sigma^{a=1,2,3}$ are the Pauli matrices.} 
    \be
    \lb{cond-Delta-D2}
   \Delta_1^1 = -\Delta_2^2, \quad \Delta_1^2 = \Delta_2^1\, .
   \ee
 
   Look now at the system  \p{dz/dx-multi}. We set $z = z_{(0)} + \delta z$. The equations acquire the form
      \be
  \frac \partial {\partial x^1}  (\delta z) + i   \frac \partial {\partial x^2} (\delta z) &=& \frac 1{\sqrt{2}}  (i\Delta_1{}^1 - \Delta_1{}^2 )\, , \nn
    \frac \partial {\partial x^2} (\delta z) - i  \frac \partial {\partial x^1} (\delta z) &=& \frac 1{\sqrt{2}} ( i\Delta_2{}^1 - \Delta_2{}^2 )\, .
    \ee
   Bearing in mind \p{cond-Delta-D2}, these two equations coincide. Introducing the notation $X^{1+i2} = X^1 + i X^2$, they can be expressed as 
   \be
     \frac {\partial (\delta z)} {\partial \bar z_{(0)}} \ =\ \frac i 2 \, \Delta_1{}^{1+i2}\,,
      \ee
      which can be easily integrated on a disk. Indeed, the whole discussion applies to a particular topologically trivial 
      chart in a set of which a manifold is subdivided.

      \item The simplest nontrivial case is $D=4$. The condition $\{\Delta, I_0\} = 0$ implies 
      \be
      \lb{cond-Delta-D4}
      \Delta_1{}^1 &=& -  \Delta_2{}^2, \qquad   \Delta_1{}^2 \ = \ \Delta_2{}^1 \, , \nn
       \Delta_1{}^3 &=& -  \Delta_2{}^4, \qquad   \Delta_1{}^4 \ = \  \Delta_2{}^3 \, , \nn
        \Delta_3{}^1 &=& -  \Delta_4{}^2, \qquad   \Delta_3{}^2 \ = \  \Delta_4{}^1 \, , \nn
       \Delta_3{}^3 &=& -  \Delta_4{}^4, \qquad   \Delta_3{}^4 \ = \  \Delta_4{}^3 \, .
        \ee
       We pose $z^1 \to  z,  z^2\to  w$. A short calculation shows that, bearing the relations 
       \p{cond-Delta-D4} in mind, the equations \p{dz/dx-multi} are reduced to 
         \be 
        \lb{eqns-D4}
        \frac {\partial (\delta z)} {\partial \bar z_{(0)}} &=& \frac i 2 \, \Delta_1{}^{1+i2} \, , \nn
         \frac {\partial (\delta z)} {\partial \bar w_{(0)}} &=& \frac i 2 \, \Delta_3{}^{1+i2} \, , \nn
         \frac {\partial (\delta w)} {\partial \bar z_{(0)} }&=&  \frac i 2 \, \Delta_1{}^{3+i4} \, , \nn
           \frac {\partial (\delta w)} {\partial \bar w_{(0)}} &=& \frac i 2 \, \Delta_3{}^{3+i4} \, .
            \ee
   If $D>2$, the conditions \p{Nijen} provide nontrivial constraints. Their linearized version is
    \be
  \lb{Nijen-Del}
  \partial_P \Delta_N{}^M - \partial_N \Delta_P{}^M \ =\ (I_0)_P{}^Q  (I_0)_N{}^S \left[ \partial_Q \Delta_S{}^M - \partial_S \Delta_Q{}^M  \right] \, . 
   \ee 
   Again, bearing in mind \p{cond-Delta-D4}, one can show that, for $D=4$, out of 24 real conditions in \p{Nijen-Del}, only 4 independent real or 2 independent complex  constraints are left. The latter have a simple form 
   \be
   \lb{cond-Delta-Nijen}
    \frac {\partial} {\partial \bar z_{(0)}} \,\Delta_3{}^{1+i2} - \frac {\partial} {\partial \bar w_{(0)}} \,\Delta_1{}^{1+i2} &=& 0 \,, \nn
    \frac {\partial} {\partial \bar z_{(0)}}\, \Delta_3{}^{3+i4} - \frac {\partial} {\partial \bar w_{(0)}}\, \Delta_1{}^{3+i4} &=& 0 \,.
     \ee
   The first equation in \p{cond-Delta-Nijen} is the integrability condition for the system of the first two equations in \p{eqns-D4}. It is necessary and also  sufficient for the 
   solution of this system to exist. Indeed, it implies that the (0,1)-form 
   $$ \omega \ =\  \Delta_1{}^{1+i2} \, d\bar z_{(0)} +   \Delta_3{}^{1+i2} \, d\bar w_{(0)} $$
     is closed, $\bar \partial_0 \, \omega = 0$. Bearing in mind the trivial topology of a chart of our complex manifold that we are discussing, $\omega$ is also exact (see e.g. Theorem 6.1 in \cite{Kodaira}), which is tantamount to saying that the solution exists. The second relation in \p{cond-Delta-Nijen} is the necessary and sufficient integrability condition
     for the system of the third and fourth equations in \p{eqns-D4}.
   
   \item This reasoning can be translated to the case of higher dimensions. 
   For an arbitrary $D = 2d$, the equations \p{dz/dx-multi} are reduced, bearing in mind $I^2 = - \mathbb{1}$, to $d^2$  conditions similar to \p{eqns-D4} but with differentiation over each antiholomorphic variable $\bar z^{\bar n}_{(0)}$ for each complex function $\delta z^n$. The conditions \p{Nijen} lead to $d^2(d-1)/2$ complex constraints which represent integrability conditions of the type \p{cond-Delta-Nijen}. They imply that the forms
      \be  
     \omega_1 \ =\  \Delta_1{}^{1+i2} \, d\bar z^1_{(0)} +   \Delta_3{}^{1+i2} \, d\bar z^2_{(0)} + \ldots , \nn
    \omega_2 \ =\  \Delta_1{}^{3+i4} \, d\bar z^1_{(0)} +   \Delta_3{}^{3+i4} \, d\bar z^2_{(0)} + \ldots ,  
    \ee
     etc.
     are all closed. Due to the trivial topology of the chart, it also means that they are exact.
     
     \item Once the complex coordinates $z^n = z^n_{(0)} + \delta z^n$ satisfying the equations \p{dz/dx-multi} are found, the complex structure acquires in these new coordinates the canonical form \p{I-compl} and \p{diag-eps}. Thus, we have actually proven that  a small deformation of $I_M{}^N$ can be brought to the form \p{diag-eps} by an infitesimal diffeomorphism, provided the condition \p{Nijen} is satisfied.
     
     \item Let now $I_M{}^N(x)$ be arbitrary, not necessarily close to $I_0$ of Eq.\p{diag-eps}. Using analyticity, we expand it into a formal series in a small parameter $\alpha$:
     \be
     \lb{I-expan}
     I(x) = I_0 + \alpha I_1(x) + \alpha^2 I_2(x) + \ldots
       \ee
       Do the same for the solutions $z^n(x)$ that we are looking for:
      \be
     \lb{z-expan}
     z^n(x) = z^n_{(0)} + \alpha z^n_{(1)}(x) + \alpha^2 z^n_{(2)}(x) + \ldots
       \ee 
       The correction $\alpha z^n_{(1)}(x)$ was determined before. Let $\tilde{z}^n(x) =  z^n_{(0)} + \alpha z^n_{(1)}(x)$. As was just mentioned, the complex structure in these new coordinates has the canonical form
       \p{I-compl} up to the terms $\propto \alpha^2$. Introducing the real and imaginary parts of  $\tilde{z}^n(x)$ and calling them $\tilde{x}^M$, we may bring it to the form \p{diag-eps}. 
        
        \item Taking also into account the term $\alpha^2 I_2(x)$ in \p{I-expan}, we may express the complex structure in the new coordinates $\tilde{x}$ as
        \be
        \lb{I-expan-new}
        I(\tilde{x}) = I_0 + \alpha^2 \tilde{I}_2 (\tilde{x}) + \ \mbox{higher-order terms.}
          \ee
          Repeating the same procedure that we used to determine $z^n_{(1)}(x)$, we can now determine 
          $\tilde{z}^n_{(2)}(\tilde{x})$, from that $z^n_{(2)}(x)$, and likewise all the terms in the series \p{z-expan}.
          
          \item With the only reservation that we did not address a difficult question of the convergence of the series \p{z-expan}, the theorem is proven.
   \end{itemize}
 \end{proof}
  
\section{Supersymmetry}

 \subsection{Preliminaries}
 To begin with, we present  some basic ``superfacts'', bearing in mind a reader who is an expert in differential geometry, but may not know much about supersymmetry. We give, however, only the minimal necessary information assuming that our reader knows the basics of Grassmann algebra and, which is not so much necessary but desirable, of classical and quantum mechanics of the systems involving Grassmann dynamical variables. More details can  be found in  the review \cite{Cooper}. See especially Chap. 8.1 there.
 
 \vspace{1mm}
 The simplest supersymmetry algebra reads
 \be
 \lb{SUSY-alg}
 Q_1^2 = Q_2^2 \ =\ H, \qquad \qquad  Q_1 Q_2 + Q_2 Q_1 \ =\ 0 \, .
  \ee
  Here $H$ is the Hamiltonian and $Q_{1,2}$ are two different Hermitian operators called {\it supercharges}. As follows from \p{SUSY-alg}, they commute with $H$.  If one introduces a complex supercharge
  $ Q \ =\ (Q_1 + iQ_2)/2$,
  one can also present \p{SUSY-alg} in the form
  \be
  \lb{SUSY-alg-compl}
  Q^2 = (\bar Q)^2 = 0, \qquad \qquad  Q \bar Q + \bar Q Q \ =\ H \,.
    \ee
  
  The algebra \p{SUSY-alg} involves two supercharges and, correspondingly, is usually called the algebra of ${\cal N} = 2$ {\it supersymmetric quantum mechanics} (SQM). More complicated algebras may involve extra supercharges\footnote{The SQM systems enjoying ${\cal N} = 4$ or 
 ${\cal N} = 8$ supersymmetry are known.}  or also the momentum operators $P_j$.  The latter algebras are relevant for supersymmetric  quantum field theories.
 But in this paper we are going to discuss only the algebra \p{SUSY-alg} and also 
 still more simple ${\cal N} = 1$ supersymmetry algebra, 
   \be
   \lb{SUSY-alg-N1}
 {\cal Q}^2 = H
   \ee
   with real ${\cal Q}$. Physically, the latter is too simple to be interesting.
  After diagonalisation, one can always extract a square root of the Hamiltonian whose spectrum is bounded from below. 
  If some energies in the spectrum are negative, one just redefines $H$ by adding an appropriate positive constant. However, we will use in what follows the algebra \p{SUSY-alg-N1} and its representations as a {\it technical tool}.
 
 The algebra \p{SUSY-alg} leads to a double degeneracy of the spectrum. It also follows from \p{SUSY-alg} that the eigenvalues of the Hamiltonian are positive or zero. The doublets involving two positive energy states $|B\rangle$ and $|F\rangle$ with the properties
  \be
  \lb{SUSY-doublet}
  H|B\rangle \ =\ E|B\rangle, &\qquad \qquad& H|F\rangle \ =\ E|F\rangle\,, \nn
  Q |B\rangle \ =\ \sqrt{E} |F\rangle\,, &\qquad \qquad& Q|F\rangle \ =\ 0\, , \nn
  \bar Q |B\rangle \ =\ 0\,, &\qquad \qquad& \bar Q |F\rangle \ =\ \sqrt{E} |B\rangle
   \ee
   represent  a simple 2-dimensional irreducible representation of the algebra \p{SUSY-alg-compl}. There exist also finite-dimensional representations involving a larger even number of states, but it is easy to show that  they are all reducible. In physical language, any  set of $2n$ states  providing a representation of \p{SUSY-alg-compl}
   is split into $n$ doublets. 
   
   The only irreducible finite-dimensional representations of the algebra \p{SUSY-alg-N1} are the trivial singlets---the eigenstates of ${\cal Q}$ and $H$. 
   
   We will be interested, however, in more complicated infinite-dimensional representations of the ${\cal N} = 1$ and ${\cal N} = 2$  algebra where the supercharges and the Hamiltonian are realized as linear differential operators acting in {\it superspace}.\footnote{Well, in supersymmetric mechanical problems, we are dealing not with ``superspace'', but rather with  ``supertime'', because we do not have any space variables and spatial dependence. But we stick to the terms commonly used in the literature.}
  
  The  ${\cal N} = 1$   superspace includes time $t$ and a real Grassmann nilpotent variable $\theta$:
   $\theta^2 = 0$. 
   The supercharges and the Hamiltonian are realized as the differential operators.
     \be
\lb{QH-N1}
{\cal Q} &=& -i \left( \frac \partial {\partial \theta} + i\theta \frac \partial {\partial t} \right)\, , \nn
H &=& -i \frac \partial {\partial t}
  \ee
 
 The Hamiltonian is the generator for the time shifts. The supercharge is the generator for somewhat more complicated transformations:
   \be
\lb{tthet-trans}
 \theta &\to& \theta + \eta \, , \nn
t &\to& t + i\eta \theta 
  \ee
 with a real Grassmann parameter $\eta$. 
 
  Consider now ${\cal N} =1$ {\it superfields} (or supervariables) representing functions of $t$ and $\theta$. Due to the nilpotency of $\theta$, they can be
  presented as
  \be
  \lb{N1-superfield}
  {\cal X}(t,\theta) = x(t) + i\theta \psi(t) \, .
   \ee
   
   The ordinary real function $x(t)$ and the Grassmann-odd real function $\psi(t)$ are called the {\it components} of the superfield \p{N1-superfield}. The shifts \p{tthet-trans} induce the shift 
   \be
   \lb{shift-X}
   \delta {\cal X} \ = \ {\cal X}(t+i\eta \theta, \theta + \eta) - {\cal X}(t, \theta) \ =\ i\eta {\cal Q} {\cal X}
    \ee
   of the superfield ${\cal X}$ implying  the following shifts of its components:
   \be
   \lb{trans-comp-N1}
   \delta x(t) \ =\ i\eta \psi(t)\, , \qquad \qquad \delta \psi(t) \ =\ -\eta
   \dot{x} \, .
    \ee 
    Note that the product of two superfields is also a superfield:
  $\delta ({\cal X}_1 {\cal X}_2) =  i \eta {\cal Q}  ({\cal X}_1 {\cal X}_2)$.

  Now we introduce  the {\it covariant supersymmetric derivative} 
  \be
  \lb{covder-N1}
   {\cal D} \ =\ \frac \partial {\partial \theta} - i\theta \frac \partial {\partial t}\,. 
 \ee
 This operator is Hermitian, nilpotent and anticommutes with ${\cal Q}$. The property
  \be
  \lb{D2=dt}
  {\cal D}^2 \ =\ -i \frac {\partial}{\partial t}
   \ee
   holds.

  \begin{thm}
  {\em If ${\cal X}$ is a superfield, the same is true for ${\cal D}{\cal X}$.}
  \end{thm}
  \begin{proof}
  We have
  $$\delta({\cal D}{\cal X}) = {\cal D} \, \delta {\cal X}  =  i {\cal D} ( 
  \eta {\cal Q}  {\cal X})  = i \eta {\cal Q} ( {\cal D} {\cal X})$$
  (do not forget that $\eta$  anticommutes with ${\cal D}$). 
  \end{proof}
  We understand now why ${\cal D}$ is called the {\it covariant} derivative. In the same way as the covariant derivative in Riemannian geometry makes a tensor out of a tensor, the derivative \p{covder-N1} makes a superfield out of a superfield.

  The superfield \p{N1-superfield} with its transformation law \p{trans-comp-N1}  defines an infinite-dimensional representation of the algebra \p{SUSY-alg-N1}.
  But it is a {\it reducible} representation. Indeed, one can now impose the constraint of reality $\bar {\cal X} = {\cal X}$. A real superfield stays real under the variation \p{shift-X}.
  
  \vspace{1mm}
  
   ${\cal N} = 2$ superspace and the ${\cal N} = 2$ superfields are defined in a similar manner.
  The superspace now includes time $t$ and a {\it complex} Grassmann anticommuting variable $\theta$:
   $\theta^2 = \bar \theta^2 = \{\theta, \bar \theta \}_+ = 0$. 
   The supertransformations are 
    \be
\lb{tthet-trans-N2}
 \theta &\to& \theta + \epsilon \, , \nn
\bar \theta &\to& \bar \theta + \bar \epsilon \, , \nn
t &\to& t + i(\epsilon \bar \theta  + \bar \epsilon \theta ) 
  \ee
with complex Grassmann $\epsilon$. 
   These transformations are generated by a complex supercharge $Q$ and its Hermitian conjugate: 
     \be
\lb{Q-gen-N2}
Q  &=& - \frac i {\sqrt{2}}  \left( \frac \partial {\partial \theta} + i \bar \theta \frac \partial {\partial t} \right) \, , \nn
\bar Q  &=& - \frac i {\sqrt{2}}  \left( \frac \partial {\partial \bar \theta} + i  \theta \frac \partial {\partial t} \right)  \nn
 \ee
 [the factor $1/\sqrt{2}$ is added to ensure the validity of \p{SUSY-alg-compl}].
   A generic ${\cal N} = 2$ superfield reads
  \be
  \lb{gen-superfield}
  \Phi(t,\theta,\bar\theta) \ =\ z(t) +  i \theta   \chi(t)  +i  \bar \theta  \lambda(t)  +  \theta \bar \theta  F(t)
   \ee
   with Grassmann-even complex $z(t)$ and $F(t)$ and Grassmann-odd 
   complex $\chi(t)$ and $\lambda(t)$. The supersymmetric variation of $\Phi$ reads 
   \be
   \delta \Phi \ =\ i\sqrt{2} (\epsilon Q + \bar \epsilon \bar Q) \Phi \, .
    \ee
    The covariant supersymmetric derivatives which are nilpotent and anticommute with $Q$ and $\bar Q$ are 
 \be
\lb{covder-N2}
D &=&  \frac \partial {\partial \theta} - i \bar \theta \frac \partial {\partial t} \, , \nn
\bar D &=& -\frac \partial {\partial \bar \theta} + i \theta \frac \partial {\partial t} \, .
  \ee
  The operator $i\bar D$ is the Hermitian conjugate of $iD$.  If $\Phi$ is a superfield, then $D\Phi$ and $\bar D \Phi$ are also superfields.
  
  The superfield \p{gen-superfield} defines an infinite-dimensional representation of the algebra \p{SUSY-alg-compl}.
  This representation is reducible. Two different irreducible representations are obtained after imposing the constraints:
   \begin{itemize}
   \item The {\it reality} constraint $\bar \Phi = \Phi$. If $\Phi$ is real, the variation $\delta \Phi$ is also real.
   \item The {\it chirality} constraints  $D \Phi = 0$ or $\bar D \Phi = 0$. Again, if $ D \Phi$ vanishes, so does $D \delta \Phi$, and the same for $\bar D$.
   Note that if $\bar D Z = 0$, then $D \bar Z = 0$. We will call $Z$ a {\it left} chiral superfield and $\bar Z$ a {\it right} chiral superfield.\footnote{The terms ``left'' and ``right'' have a physical origin which is irrelevant for us here.}
   \end{itemize}
    
    In what follows, we will not be interested in the real ${\cal N}=2$ superfields, but exclusively in the chiral ones.

   For a chiral superfield, the component expansion \p{gen-superfield} can be simplified if one introduces  ``left'' and ``right'' times: 
    $$ t_L \ =\ t - i\theta \bar \theta\,, \qquad \qquad t_R \ =\ t + i\theta \bar \theta\,.$$
    The supersymmetric variation of $t_L$ depends only on $\theta$, $\delta t_L = 2i \bar \epsilon \theta$, and the supersymmetric variation 
    of $t_R$ depends only on $\bar \theta$.
  
 The set of coordinates $(t_L, \theta)$ describes the  {\em holomorphic  chiral ${\cal N} = 2$ superspace} and the set  $(t_R, \bar\theta)$ describes the  {\it antiholomorphic chiral  ${\cal N} = 2$ superspace.}

Then, if $\bar D Z = 0$, we may write
\be
\lb{Z-chiral-LR}
Z \ =\ { Z}(t_L, \theta) &=& z(t_L) + i \sqrt{2}\, \theta \, \chi (t_L) \, , \nn
\bar Z \ =\ \bar { Z}(t_R,  \bar\theta) &=& \bar z(t_R)  + i \sqrt{2} \, \bar \theta \, \bar \chi (t_R) \,.
 \ee
  The components of a left chiral superfield are transformed as  
 \be
\lb{trans-Land1}
 \delta z  \ = \  i \sqrt{2}   \epsilon \, \chi \, , \qquad \qquad    \delta \chi \ = \ -  \sqrt{2} \bar \epsilon  \, \dot{z} \, .
  \ee
  Let us pose now
  \be
  z = \frac{x_1 + ix_2}{\sqrt{2}} \,, \qquad \chi = \frac{\psi_1 + i\psi_2}{\sqrt{2}} \,, \qquad \epsilon = \frac {\eta + i\tilde \eta}{\sqrt{2}} \, . 
  \ee
  Suppose first that $\epsilon$ is real, $\tilde \eta = 0$. Then we derive
    \be 
    \lb{trans-eta}
    \delta x_1 = i \eta \psi_1, &\qquad \qquad \qquad&  \delta \psi_1 = - \eta \dot{x}_1  \, , \nn
  \delta x_2 = i \eta \psi_2, &\qquad \qquad \qquad&  \delta \psi_2 = - \eta \dot{x}_2  \, .
   \ee
   We see that the components $(x_1, \psi_1)$ are not mixed with the components $(x_2, \psi_2)$; each set is transformed in the same way as the components of an ${\cal N} = 1$ superfield [see Eq. \p{trans-comp-N1}]! In other words, the representation $Z$ is an irreducible representation of the ${\cal N} = 2$ superalgebra, but it can also be thought of as a {\it reducible} representation of  ${\cal N} = 1$ superalgebra realized by the transformations \p{trans-Land1} with real $\epsilon$. When going down from ${\cal N} = 2$  to ${\cal N} = 1$, the chiral superfield $Z$ is split into two real superfields ${\cal X}_1$ and  ${\cal X}_2$.
   To see it quite explicitly, substitute  $\theta = (\theta_1 + i\theta_2)/\sqrt{2}$ in \p{Z-chiral-LR}. Then $t_L = t + \theta_2 \theta_1$. We derive
    \be
    \lb{N2-via-N1}
    Z \ =\ \frac 1 {\sqrt{2}} \left\{ {\cal X}_1(t, \theta_1) + i {\cal X}_2(t, \theta_1) + 
    i\theta_2 [{\cal D} {\cal X}_1(t, \theta_1) + i {\cal D} {\cal X}_2(t,\theta_1)] \right\} \, .
    \ee

   Look now at the transformations \p{trans-Land1} when $\epsilon = i\tilde \eta/\sqrt{2}$ is imaginary. We obtain
     \be
     \lb{trans-tilde-eta}
     \tilde \delta x_1 = -i \tilde \eta \psi_2, &\qquad \qquad \qquad& \tilde \delta \psi_1 = -\tilde \eta \dot{x}_2, \nn
      \tilde \delta x_2 = i \tilde \eta \psi_1, &\qquad \qquad \qquad& \tilde \delta \psi_2 = \tilde \eta \dot{x}_1 
       \ee
       or  in a compact form:
    \be
    \lb{trans-tilde-XA}
    \tilde \delta {\cal X}_A \ =\ \tilde \eta  \, \varepsilon_{AB} \, {\cal D} {\cal X}_B 
     \ee  
     [with $\varepsilon$ defined as in \p{eps-def}].

  The generators of the transformations \p{trans-eta} and \p{trans-tilde-eta}  obey the  algebra \p{SUSY-alg}. Indeed, 
 \begin{itemize}
 \item It is rather evident that the transformations \p{trans-eta} and (\ref{trans-tilde-eta}, \ref{trans-tilde-XA}) commute.
Indeed, $\delta {\cal X}_A$ is a superfield, and hence $\delta (\tilde \delta {\cal X}_A)$ and $\tilde \delta (\delta {\cal X}_A)$ coincide, having both the form \p{shift-X} with
${\cal X}$ replaced by $\tilde\delta {\cal X}_A$.
  A corollary of this is the vanishing of the anticommutator
  ${\cal Q}\tilde {\cal Q} + \tilde {\cal Q} {\cal Q}$ of the corresponding quantum supercharges.
  \item Bearing in mind \p{D2=dt}, the Lie bracket of two different tilde-transformations reads
   \be
   (\tilde\delta_1    \tilde \delta_2 -  \tilde\delta_2    \tilde \delta_1){\cal X}_A = -2i \tilde\eta_1 \tilde\eta_2 \dot {\cal X}_A \, ,
    \ee
    which is tantamout to saying that $\tilde {\cal Q}^2$ coincides with the Hamiltonian (the generator of time shifts).
  \end{itemize}

   \subsection{NN theorem: supersymmetric interpretation}
   
   The tensor $\varepsilon_{AB}$ entering \p{trans-tilde-XA} can be interpeted as a $2\times 2$ block in the flat complex structure \p{diag-eps}. The components $x_A$ of the superfields ${\cal X}_A$ can be interpreted as the flat Cartesian coordinates. Suppose now that we have $2d$ ${\cal N} = 1$ superfields ${\cal X}^M$. One of the supersymmetries follows from the transformations of the superspace coordinates as in \p{trans-eta}:
    \be
    \lb{trans-eta-M}
   \delta x^M = i \eta \psi^M, &\qquad \qquad \qquad&  \delta \psi^M = - \eta \dot{x}^M  \, .
    \ee
    Looking  for a generalization of \p{trans-tilde-XA}, we anticipate the presence of the second supersymmetry,
     \be
    \lb{trans-tilde-XM}
    \tilde \delta {\cal X}^M \ =\ \tilde \eta  \, I_N{}^M({\cal X}^P) \, {\cal D} {\cal X}^N \, ,
     \ee  
   where 
    \be
    \lb{I2=-1}
   I^2 = - \mathbb{1}\,,
   \ee
    and ask: {\it under what conditions is it possible?} Under what conditions do
 the generators of the transformations \p{trans-eta-M}
   and \p{trans-tilde-XM} obey the algebra \p{SUSY-alg}?
   
   \begin{thm}
  {\em  The algebra \p{SUSY-alg} holds iff the Nijenhuis tensor \p{Nijen} vanishes. } \lb{NN2}
    \end{thm}
    \begin{proof}
   The Lie bracket $[\delta, \tilde \delta]$ vanishes by the same reason as in the flat case treated before: the transformation $\delta$ mixes the components of each multiplet, while the transformation $\tilde \delta$ mixes different superfields and does not bother much about their internal structure.
   Thus, we only need to explore the Lie bracket $(\tilde\delta_1    \tilde \delta_2 -  \tilde\delta_2    \tilde \delta_1){\cal X}^M$.
   
   Note first that
  $$ \tilde \delta ({\cal D} {\cal X}^N) = {\cal D}(\tilde \delta {\cal X}^N) \ =\ -\tilde \eta {\cal D} (I_L{}^N {\cal D} {\cal X}^L) \ =\ 
   -\tilde\eta (\partial_K I_L{}^N) {\cal D} {\cal X}^K  {\cal D} {\cal X}^L + i \tilde \eta I_L{}^N \dot{\cal X}^L \,.$$
   The commutator  of two transformations \p{trans-tilde-XM} is then derived to be
   \be
\lb{com-s-X}
\left(\tilde\delta_1\tilde\delta_2- \tilde\delta_2\tilde\delta_1 \right) {\cal X}^M &=&
2i\tilde \eta_1\tilde \eta_2 (I^2)_K{}^{M} \dot{\cal X}^K \\
&&
-2\tilde \eta_1\tilde \eta_2\Big[I_K{}^{L}\left(\partial_L I_N{}^{M}\right) +\left(\partial_N I_K{}^{L}\right)I_L{}^{M}\Big]\,{\cal D}{\cal X}^K {\cal D}{\cal X}^N \, . \nonumber
\ee
If we want it to coincide with $-2i\tilde \eta_1\tilde \eta_2\, \partial_t {\cal X}^M$  [as is dictated by Eq.\p{SUSY-alg}] the conditions \p{I2=-1} as well as
\begin{equation}
\lb{eq-intr}
\left(\partial_L I_{[N}{}^{M}\right)I_{K]}{}^{L} +\left(\partial_{[N} I_{K]}{}^{L}\right)I_L{}^{M} \ =\ 0
\end{equation}
follow.  Using again \p{I2=-1} and flipping the derivative in the second term, the L.H.S. of Eq. \p{eq-intr} can be brought into the form
\p{conven-Nijen}. The condition \p{Nijen} follows.
   
  \end{proof} 
  
  Thus, the condition ${\cal N}_{MN}{}^K = 0$ is necessary and sufficient for ${\cal N} = 2$ supersymmetry associated with the given complex structure to hold. But the NN theorem is formulated differently: it affirms that the condition \p{Nijen} is necessary and sufficient for the existence of complex coordinates.
  
  Well, as far as necessity is concerned, the equivalence of Theorems \ref{NN1} and \ref{NN2} is rather clear. Suppose that complex coordinates $z^n$ exist. But then each such coordinate can be upgraded to a complex chiral superfield $Z^n$ whose components are transformed under supersymmetry as in \p{trans-Land1}. Each superfield $Z^n$ can be expressed via a pair of ${\cal N}=1$ real superfields as in \p{N2-via-N1}.
  The complex structure tensor $I_M{}^N$ has in this case the form \p{diag-eps} and does not depend on the coordinates. The tensor ${\cal N}_{MN}{}^K$ vanishes automatically.
  
  Now, if the Nijenhuis tensor vanishes, we know from Theorem \ref{NN2} that the algebra of ${\cal N} = 2$ supersymmetry holds. The set of 
  $2d$ superfields ${\cal X}^M$ is an infinite-dimensional representation of this algebra. Then the sufficiency of \p{Nijen} means that, for $d > 1$, this representation is 
  {\it reducible} and can be decomposed in a direct sum of $d$ irreducible representations realized by the components of the chiral complex superfields $Z^n$.
  
  This latter statement looks very natural, it is widely used by physicists, but I am not aware of its independent proof. The only known proof of this fact is the proof of the sufficiency part of the NN theorem that we outlined in Sect. 2 and  that does not resort to supersymmetric description. 
  
  \subsubsection{Invariant actions}
  Up to now, when talking about the supersymmetric aspects of the NN theorem, we stayed at the purely algebraic level, having discussed only the algebras \p{SUSY-alg}, \p{SUSY-alg-N1} and their representations. A reader-mathematician may stop reading this paper at this point.
  
  But, when a physicist thinks of a symmetry, s/he is always interested in {\it dynamical systems} that enjoy these symmetries. An industrial method to find supersymmetric dynamical systems is based on the following theorem:
   \begin{thm}
   \lb{sup-act}
 {\em Let ${\cal X}(t,\theta)$ be an ${\cal N} = 1$ superfield that vanishes at $t = \pm \infty$. Then the integral (associated with the physical action)
    
  \be
  \lb{act-N1} 
  S \ =\ \int d\theta \int_{-\infty}^\infty dt \, {\cal X} 
   \ee 
   is invariant under transformations \p{tthet-trans}.}
  \end{thm}
  Here  the symbol
 $\int \! d\theta$ is the Berezin integral, 
 \be
 \lb{Ber-def}
 \int d\theta {\cal X}  \ \equiv \ \frac \partial {\partial \theta} {\cal X} \,.
  \ee

  \begin{proof}
  We have 
  $$ \delta S \ =\ \int d\theta  \int_{-\infty}^\infty dt \, \delta {\cal X} \ =\ -\epsilon \int d\theta \int_{-\infty}^\infty dt \left( \frac \partial {\partial \theta} + i\theta \frac \partial {\partial t} \right) {\cal X} \,. $$
  The first term vanishes due to the definition \p{Ber-def} and the Grassmannian nature of $\theta$. The second term vanishes due to the condition
  ${\cal X}(\pm \infty, \theta) = 0$.
  \end{proof}
  Obviously, the same property holds for  the integral 
  \be
  \lb{act-N2} 
  S \ =\ \int d\bar\theta d\theta \int_{-\infty}^\infty \, dt \, \Phi 
   \ee
  of a ${\cal N} = 2$ superfield $\Phi$.
  
  The superfield ${\cal X}$  in Eq. \p{act-N1} and the superfield $\Phi$ in   Eq. \p{act-N2} can be constructed out of certain basic superfields by multiplications, time differentiations and covariant differentiations with the operator ${\cal D}$ in the ${\cal N} = 1$ case and with the operators $D$ and $\bar D$ in the ${\cal N} = 2$ case. In particular, one can write
   \be
\lb{SN2-Dolb}
S \ =\  \frac 14 \int d\bar \theta d\theta dt \, h_{m\bar n} (Z^k, \bar Z^{\bar k}) \,\bar D 
\bar Z^{\bar n}(t_R) DZ^m(t_L)   \, ,
  \ee
  where $Z^{k=1,\ldots,d}$ are left chiral superfields and $h_{m\bar n}$ is Hermitian. Substituting there the expansions \p{Z-chiral-LR}, not forgetting to expand over $\theta$ and $\bar\theta$ also $t_{L,R} = t \mp i\theta\bar\theta$ and performing the integral over $d\bar\theta d\theta dt$, one can derive the following expression for the Lagrangian:
  \be
\lb{L-compl-N2}
L  \ =\   h_{m\bar n}(z, \bar z)  \dot{z}{\,}^m\dot{\bar z}{\,}^{\bar n} \ + \ {\rm terms\ including\ superpartners}\  \chi^m(t)
\ee
We can now interpret  $z^m$ and $\bar z^{\bar m}$ as the coordinates on a complex manifold with the metric $h_{m\bar n}(z, \bar z)$. The displayed term of the Lagrangian can be interpreted as the kinetic energy of a particle with unit mass moving along the manifold. The dynamical system describing such a motion is called {\it sigma model}. And the whole Lagrangian [due to  Theorem \ref{sup-act}, the corresponding action is invariant under \p{trans-Land1}] represents its supersymmetric version.

The {\it same} dynamical system can also be described in the ${\cal N} = 1$ superfield language. Consider the action \cite{Coles}
\be
\lb{SgenN1}
S \ =\ \frac i2 \int  d\theta dt \, g_{MN}({\cal X}) \, {\dot {\cal X}}^M {\cal D} {\cal X}^N  \, ,
  \ee
   This is not a most general form. The action \p{SgenN1} describes (under the condition that ${\cal N}_{MN}{}^K$ vanishes) only the K\"ahler manifolds; to describe generic complex manifolds, one should add  an extra term. But we do not want to plunge into too much details here, addressing an interested reader to Sect. 4 of Ref. \cite{HKT-nonlin}.
   
  After integration over $d\theta dt$, we obtain the Lagrangian 
   \be
   \lb{L-real-N2}
  L \ =\ \frac 12 g_{MN} \,\dot{x}^M \dot{x}^N  \ + \  {\rm terms\ including\ superpartners}\  \psi^M(t)\,,
   \ee
   i.e. $g_{MN}$ has the meaning of the real metric.
   
  By construction, the action \p{SgenN1} is invariant under ${\cal N} = 1$ transformations, but it is also invariant under the extra supersymmetry transformations
  \p{trans-tilde-XM} provided the conditions \p{I2=-1}, \p{Nijen} {\it and} the condition $I_{MN} = -I_{NM}$ hold. 
  
  Note that, to relate $I_{MN}$ to $I_M{}^N$, we need the metric. The notion of metric was {\it not} used in the proof of Theorem \ref{NN2} or Theorem \ref{NN1}, 
  which thus hold also for non-metric manifolds. Indeed, the equation system \p{dz/dx-multi} for the complex coordinates has solutions provided the condition \p{Nijen} is fulfilled even when $I_{MN} \neq -I_{NM}$.  But we need the metric for the physical applications. And then the condition of the antisymmetry of $I_{MN}$ should be imposed.
  
  The equations of motion that follow from the Lagrangian \p{L-real-N2} describe classical supersymmetric dynamics. The Legendre transformation of \p{L-real-N2} gives us the classical Hamiltonian from which the quantum Hamiltonian can be derived. The quantum system has the same symmetry as the classical one. If we are dealing with ${\cal N} =2$ supersymmetry, a pair of Hermitially conjugate supercharges satisfying the algebra \p{SUSY-alg-compl} exist. This guarantees the two-fold degeneracy of all positive energy states as in \p{SUSY-doublet}.

\begin{acknowledgement}
I am indebted to G. Carron, G. Papadopoulos and A. Rosly for illuminating discussions.
\end{acknowledgement}

\end{document}